\def\arxiv{1}
\ifnum\arxiv=1
    \documentclass[11pt]{article}
\else
    \documentclass[a4paper, USenglish, authorcolumns]{lipics-v2019}
\fi

\usepackage{amsmath}
\usepackage{amssymb}
\usepackage{amsthm}
\usepackage{bbm}
\usepackage{enumitem}
\usepackage{float}

\usepackage{hyperref}
\usepackage{makecell}
\usepackage{mdframed}
\usepackage{tikz}
\usepackage[noend]{algorithm2e}
\usepackage{multirow}

\ifnum\arxiv=1
    \usepackage[margin=1in]{geometry}
    \usepackage{hyperref}
    \hypersetup{
        colorlinks=true,
        linkcolor=[rgb]{0 0 .6},
        citecolor=[rgb]{0 0 .6}
    }
    \theoremstyle{plain}
    
\else
\fi
\newtheorem{thm}{Theorem}[section]
    \newtheorem{clm}[thm]{Claim}
    \newtheorem{lem}[thm]{Lemma}
    \newtheorem{coro}[thm]{Corollary}
    \theoremstyle{definition}
    \newtheorem{defn}[thm]{Definition}

\newcommand{\N}{\mathbb{N}}
\newcommand{\R}{\mathbb{R}}
\newcommand{\Z}{\mathbb{Z}}

\newcommand{\bD}{\mathbf{D}}
\newcommand{\bU}{\mathbf{U}}

\newcommand{\cA}{\mathcal{A}}

\newcommand{\cM}{\mathcal{M}}

\newcommand{\cP}{\mathcal{P}}

\newcommand{\cR}{\mathcal{R}}
\newcommand{\cS}{\mathcal{S}}

\newcommand{\cX}{\mathcal{X}}
\newcommand{\cY}{\mathcal{Y}}
\newcommand{\cZ}{\mathcal{Z}}
\newcommand{\tc}{\tilde{c}}

\newcommand{\Ber}{\mathrm{Ber}}
\newcommand{\Bin}{\mathrm{Bin}}

\newcommand{\Pois}{\mathrm{Pois}}

\newcommand{\poly}{\mathrm{poly}}

\newcommand{\hist}{\mathrm{hist}}
\newcommand{\zc}{\mathrm{zsum}}

\newcommand{\eps}{\varepsilon}
\newcommand{\zo}{\{0,1\}}

\ifnum\arxiv=1
    \title{Separating Local \& Shuffled \\Differential Privacy via Histograms\\}
    \author{
    Victor Balcer\footnote{School of Engineering \& Applied Sciences, Harvard University. Supported by NSF grant CNS-1565387.}~ and Albert Cheu\footnote{Khoury College of Computer Sciences, Northeastern University. Supported by NSF grants CCF-1718088, CCF-1750640, and CNS-1816028.}}
    \date{\today}
\else
    \title{Separating Local \& Shuffled \\Differential Privacy via Histograms}
    
    \titlerunning{Separating Local \& Shuffled D.P.}
    
    \author{Victor Balcer}
    {School of Engineering \& Applied Sciences, Harvard University, United States.}{}{}{Supported by NSF grant CNS-1565387.}
    
    \author{Albert Cheu}{Khoury College of Computer Sciences, Northeastern University, United States.}{}{}{Supported by NSF grants CCF-1718088, CCF-1750640, and CNS-1816028.}
    \authorrunning{V. Balcer and A. Cheu}
    
    \Copyright{Victor Balcer and Albert Cheu}
    
    
    
    \nolinenumbers
    
    \acknowledgements{We are grateful to Daniel Alabi and Maxim Zhilyaev for discussions that shaped the early stages of this work.
    We are also indebted to Matthew Joseph and Jieming Mao for directing us to the pointer-chasing and multi-party pointer-jumping problems. Finally, we thank Salil Vadhan for editorial comments and providing a simpler construction for Claim \ref{clm:pure-removal-multi}.}
    
    \EventEditors{Yael Tauman Kalai, Adam D. Smith, and Daniel Wichs}
    \EventNoEds{3}
    \EventLongTitle{1st Conference on Information-Theoretic Cryptography (ITC 2020)}
    \EventShortTitle{ITC 2020}
    \EventAcronym{ITC}
    \EventYear{2020}
    \EventDate{June 17--19, 2020}
    \EventLocation{Boston, MA, USA}
    \EventLogo{}
    \SeriesVolume{163}
    \ArticleNo{1}
\fi

\begin{document}
\maketitle

\begin{abstract}
    Recent work in differential privacy has highlighted the \emph{shuffled model} as a promising avenue to compute accurate statistics while keeping raw data in users' hands.
    We present a protocol in this model that estimates histograms with error \emph{independent of the domain size}.
    This implies an arbitrarily large gap in sample complexity between the shuffled and local models.
    On the other hand, we show that the models are equivalent when we impose the constraints of pure differential privacy and single-message randomizers.
\end{abstract}

\section{Introduction}
The \emph{local model} of private computation has minimal trust assumptions: each user executes a randomized algorithm on their data and sends the output \emph{message} to an analyzer \cite{KLNRS08,Warner65,EGS03}.
While this model has appeal to the users---their data is never shared in the clear---the noise in every message poses a roadblock to accurate statistics. For example, a locally private $d$-bin histogram has, on some bin, error scaling with $\sqrt{\log d}$. But when users trust the analyzer with their raw data (the \emph{central model}), there is an algorithm that achieves error independent of $d$ on every bin.


Because the local and central models lie at the extremes of trust, recent work has focused on the intermediate \emph{shuffled model} \cite{BittauEMMRLRKTS17, CSU+19}. In this model, users execute ranomization like in the local model but now a trusted \emph{shuffler} applies a uniformly random permutation to all user messages before the analyzer can view them. The anonymity provided by the shuffler allows users to introduce less noise than in the local model while achieving the same level of privacy. This prompts the following questions:

\begin{quote}
\begin{center}
\textit{
In terms of accuracy, how well separated is the shuffled model from the local model?\\
\vspace{.3ex}How close is the shuffled model to the central model?}
\end{center}
\end{quote}

\subsection{Our Results}
In Section \ref{sec:hist}, we provide a new protocol for histograms in the shuffled model. To quantify accuracy, we bound the \emph{simultaneous error}: the maximum difference over all bins between each bin's estimated frequency and its true frequency in the input dataset.

\begin{thm}[Informal]
For any $\eps<1$ and $\delta=o(1/n)$, there exists a shuffled protocol that satisfies $(\eps, \delta)$-differential privacy and reports a histogram with simultaneous error $O(\log(1/\delta)/ (\eps^2 n) )$ with constant probability.
\end{thm}

For comparison, \cite{CSU+19} give a protocol with error $O(\sqrt{\log d \cdot \log 1/\delta} / (\eps n))$. Our protocol has smaller error when $\log (1/\delta) = o(\log d)$. In the natural regime where $\delta= \Theta(\mathit{poly}(1/n))$, that condition is satisfied when $\log n = o(\log d)$.
An example for this setting would be a dataset holding the browser home page of each user.
The data universe could consist of all strings up to a certain length which would far exceed the number of users.

In Section \ref{sec:applications}, we show that the histogram protocol has strong implications for the \emph{distributional} setting. Here, the rows of the dataset are independently drawn from a probability distribution. We focus on the \emph{sample complexity}, which is the number of samples needed to identify or estimate some feature of the distribution. We prove that the separation in sample complexity between the local and shuffled models can be made arbitrarily large:
\begin{thm}[Informal]
\label{thm:sep-non-int}
There is a distributional problem where the sample complexity in the local model scales with a parameter of the problem, but the sample complexity in the shuffled model is independent of that parameter.
\end{thm}
We also show that there is a distributional problem which requires polynomially more samples in the \emph{sequentially interactive} local model than in the shuffled model. This is done by reusing the techniques to prove Theorem \ref{thm:sep-non-int}.

A natural conjecture is that there are progressively weaker versions of Theorem \ref{thm:sep-non-int} for progressively constrained versions of the model. In Section \ref{sec:pure}, we prove that the shuffled model collapses to the local model when constraints are too strong:
\begin{thm}[Informal]
\label{thm:need-delta}
For every single-message shuffled protocol that satisfies pure differential privacy, there is a local protocol with exactly the same privacy and sample complexity guarantees.
\end{thm}

\subsection{Related Work}
Table \ref{tab:main-results} presents our histogram result alongside existing results for the problem---all previous bounds on simultaneous error in the shuffled model depend on $d$.
Although we focus on simultaneous error, error metrics focusing on per-bin error are also used in the literature such as mean squared error (MSE) and high probability confidence intervals on each bin.
When considering these alternative metrics or when $d$ is not large, other histogram protocols may outperform ours (see e.g. \cite{WXD+19}).
\begin{table}
\renewcommand*{\arraystretch}{1.8}
\caption{Comparison of results for the histogram problem. To simplify the presentation, we assume constant success probability, $\eps<1$, $\delta < 1/\log d$ for results from \cite{GGK+19}, and $e^{-O(n\eps^2)} \le \delta < 1/n$ for our result.\vspace{1ex}}
    \label{tab:main-results}
    \centering
    \begin{tabular}{cccc}
     Model & Simultaneous Error & No. Messages per User & Source\\ \Xhline{2.5\arrayrulewidth}
     Local &
     $ \Theta\!\left( \frac{1}{\eps \sqrt{n\,}}\cdot\sqrt{\log d\,} \right)$ &
     1 &
     \cite{BS15}\vspace{2.2pt} \\ \hline
     \multirow{4}{*}{Shuffled} &
     $ O\!\left(\frac{1}{\eps n} \cdot\sqrt{\log d \cdot\log \frac{1}{\delta}\,} \right)$ &
     $O(d)$ &
     \cite{CSU+19} \\ 
      & 
     $O\!\left(\frac{1}{\eps n}\sqrt{\log^3 d \cdot \log (\frac{1}{\delta} \log d)}\,\right)$ &
     $O\!\left(\frac{1}{\eps^2}\log^3 d \cdot \log (\frac{1}{\delta} \log d)\right)$ w.h.p. &
     \cite{GGK+19} \\
      &
     $O\!\left(\frac{\log d}{n} + \frac{1}{\eps n} \cdot\sqrt{\log d \cdot\log \frac{1}{\eps \delta} \,}\right)$ &
     $O\!\left(\frac{1}{\eps^2}\log\frac{1}{\eps\delta}\right)$ &
     \cite{GGK+19} \\
      &
     $O\!\left(\frac{1}{\eps^2 n} \log \frac{1}{\delta} \right)$ &
     $O(d)$ &
     [Theorem \ref{thm:hist}]\vspace{2.2pt} \\ \hline
     Central &
     $ \Theta\!\left(\frac{1}{\eps n} \min\left(\log d, \log \frac{1}{\delta}\right) \right)$ &
     N/A &
     \cite{DMNS06,BNS16,BBKN14,HT10} \\ 
     \end{tabular}
\end{table}

Quantitative separations between the local and shuffled models exist in the literature \cite{CSU+19, BBGN19-2, GGK+19, GGK+20}. As a concrete example, \cite{CSU+19} implies that the sample complexity of Bernoulli mean estimation in the shuffled model is $O(1/\alpha^2 + \log(1/\delta)/(\alpha\eps))$. In contrast, \cite{BNO08} gives a lower bound of $\Omega(1/\alpha^2\eps^2)$ in the local model.

Prior work have shown limits of the shuffled model, albeit under communication constraints. The first set of results follow from a lemma in \cite{CSU+19}: a single-message shuffled protocol implies a local protocol with a weaker differential privacy guarantee. Specifically, if the shuffled protocol obeys $(\eps,\delta)$-differential privacy, then the local protocol obeys $(\eps+\ln n, \delta)$-differential privacy. Lower bounds for the local model can then be invoked, as done in \cite{CSU+19,GGK+19}.

Another class of lower bound comes from a lemma in \cite{GGK+20}: when a shuffled protocol obeys $\eps$-differential privacy and bounded communication complexity, the set of messages output by each user is insensitive to their personal data. Specifically, changing their data causes the set's distribution to change by $\leq 1-2^{-O_\eps(m^2\ell)}$ in statistical distance, where $m$ denotes number of messages and $\ell$ the length of each message. This is strong enough to obtain a lower bound on binary sums. 

The amplification-by-shuffling lemmas in \cite{BBGN19,EFMRTT19} show that uniformly permuting the messages generated by a local protocol improves privacy guarantees: an $\eps$-private local protocol becomes an $(\eps',\delta)$-private shuffled protocol where $\eps'\ll \eps$ and $\delta>0$. One might conjecture weaker versions of these lemmas where $\delta=0$ but Theorem \ref{thm:need-delta} eliminates that possibility.

\section{Preliminaries}
We define a \emph{dataset} $\vec{x} \in \cX^n$ to be an ordered tuple of $n$ rows where each row is drawn from a \emph{data universe} $\cX$ and corresponds to the data of one user.
Two datasets $\vec{x},\vec{x}\,' \in \cX^n$ are considered \emph{neighbors} (denoted as $\vec{x} \sim \vec{x}\,'$) if they differ in exactly one row.

\begin{defn}[Differential Privacy \cite{DMNS06}]
An algorithm $\cM: \cX^n \rightarrow \cZ$ satisfies \emph{$(\eps, \delta)$-differential privacy} if
\begin{align*}
\forall \vec{x} \sim \vec{x}\,' ~~\forall T \subseteq \cZ ~~~ \Pr[\cM(\vec{x}) \in T] \le e^\eps \cdot \Pr[\cM(\vec{x}\,') \in T] + \delta.
\end{align*}

We say an $(\eps, \delta)$-differentially private algorithm satisfies \emph{pure differential privacy} when $\delta = 0$ and \emph{approximate differential privacy} when $\delta > 0$.
For pure differential privacy, we may omit the $\delta$ parameter from the notation.
\end{defn}

\begin{defn}[Local Model \cite{KLNRS08}]
A protocol $\cP$ in the \emph{(non-interactive\footnote{The literature also includes interactive variants; see \cite{JMNR19} for a definition of \emph{sequential} and \emph{full} interactivity.}) local model} consists of two randomized algorithms:
\begin{itemize}
    \item A \emph{randomizer} $\cR:\cX \rightarrow \cY$ that takes as input a single user's data and outputs a message.
    \item An \emph{analyzer} $\cA:\cY^* \rightarrow \cZ$ that takes as input all user messages and computes the output of the protocol.
\end{itemize}
We denote the protocol $\cP = (\cR, \cA)$. 
We assume that the number of users $n$ is public and available to both $\cR$ and $\cA$.
Let $\vec{x} \in \cX^n$.
The evaluation of the protocol $\cP$ on input $\vec{x}$ is
The evaluation of the protocol $\cP$ on input $\vec{x}$ is
\begin{align*}
\cP(\vec{x})
=  (\cA \circ\cR)(\vec{x})
= \cA(\cR(x_1), \ldots, \cR(x_n)).
\end{align*}
\end{defn}

\begin{defn} [Differential Privacy for Local Protocols]
A local protocol $\cP = (\cR, \cA)$ satisfies \emph{$(\eps, \delta)$-differential privacy} for $n$ users if its randomizer $\cR: \cX \rightarrow \cY$ is $(\eps, \delta)$-differentially private (for datasets of size one).
\end{defn}

\begin{defn}[Shuffled Model \cite{BittauEMMRLRKTS17, CSU+19}]
A protocol $\cP$ in the \emph{shuffled model} consists of three randomized algorithms:
\begin{itemize}
\item
    A \emph{randomizer} $\cR: \cX \rightarrow \cY^*$ that takes as input a single user's data and outputs a vector of messages whose length may be randomized. If, on all inputs, the probability of sending a single message is 1, then the protocol is said to be \emph{single-message}. Otherwise, the protocol is said to be \emph{multi-message}.
\item
    A \emph{shuffler} $\cS: \cY^* \rightarrow \cY^*$ that concatenates all message vectors and then applies a uniformly random permutation to (the order of) the concatenated vector. For example, when there are three users each sending two messages, there are $6!$ permutations and all are equally likely to be the output of the shuffler.
\item
    An \emph{analyzer} $\cA: \cY^* \rightarrow \cZ$ that takes a permutation of messages to generate the output of the protocol.
\end{itemize}

As in the local model, we denote the protocol $\cP = (\cR, \cA)$ and assume that the number of users $n$ is accessible to both $\cR$ and $\cA$.
The evaluation of the protocol $\cP$ on input $\vec{x}$ is
\begin{align*}
\cP(\vec{x})
= (\cA \circ \cS\circ\cR)(\vec{x})
= \cA(\cS(\cR(x_1), \ldots, \cR(x_n))).
\end{align*}
\end{defn}

\begin{defn} [Differential Privacy for Shuffled Protocols \cite{CSU+19}]
A shuffled protocol $\cP = (\cR, \cA)$ satisfies \emph{$(\eps, \delta)$-differential privacy} for $n$ users if the algorithm $(\cS\circ \cR): \cX^n \rightarrow \cY^*$ is $(\eps, \delta)$-differentially private.
\end{defn}

\medskip

We note a difference in robustness between the local and shuffled models.
A user in a local protocol only has to trust that their own execution of $\cR$ is correct to ensure differential privacy.
In contrast, a user in a shuffled protocol may not have the same degree of privacy when other users deviate from the protocol.

\medskip

For any $d \in \N$, let $[d]$ denote the set $\{1, \ldots, d\}$.
For any $j \in [d]$, we define the function $c_j: [d]^n \rightarrow \R$ as the normalized count of $j$ in the input:
\begin{align*}
c_j(\vec{x}) = (1/n)\cdot|\{i \in [n] \,:\, x_i = j\}|.
\end{align*}
We use \emph{histogram} to refer to the vector of normalized counts $(c_1(\vec{x}), \dots, c_d(\vec{x}))$. 
For measuring the accuracy of a histogram protocol $\cP: [d]^n \rightarrow \R^d$, we use the following metrics:

\begin{defn}
A histogram protocol $\cP: [d]^n \rightarrow \R^d$ has \emph{$(\alpha, \beta)$-per-query accuracy} if
\begin{align*}
\forall \vec{x} \in [d]^n ~~ \forall j \in [d] ~~~ \Pr[|\cP(\vec{x})_j - c_j(\vec{x})| \le \alpha] \ge 1 - \beta.
\end{align*}
\end{defn}

\begin{defn}
A histogram protocol $\cP: [d]^n \rightarrow \R^d$ has \emph{$(\alpha, \beta)$-simultaneous accuracy} if
\begin{align*}
 \forall \vec{x} \in [d]^n ~~~ \Pr[\forall j \in [d] ~~~ | \cP(\vec{x})_j - c_j(\vec{x}) | \le \alpha] \ge 1 - \beta.
\end{align*}
\end{defn}

\section{The Power of Multiple Messages for Histograms}
\label{sec:hist}

In this section, we present an $(\eps,\delta)$-differentially private histogram protocol in the shuffled model whose simultaneous error does not depend on the universe size. We start by presenting a private protocol for releasing a binary sum that always outputs 0 if the true count is 0 and otherwise outputs a noisy estimate. The histogram protocol uses this counting protocol to estimate the frequency of every domain element. Its simultaneous error is the maximum noise introduced to the nonzero counts. There are at most $n$ such counts.

For comparison, a protocol in \cite{CSU+19} adds independent noise to all counts without the zero-error guarantee. The simultaneous error is therefore the maximum noise over \emph{all} $d$ counts, which introduces a $\log d$ instead of a $\log n$ factor.

\subsection{A Two-Message Protocol for Binary Sums}

In the protocol $\cP^\zc_{\eps,\delta}$ (Figure \ref{fig:zsum}), each user reports a vector whose length is the sum of their data and a Bernoulli random variable. The contents of each vector will be copies of 1. Because the shuffler only reports a uniformly random permutation, the observable information is equivalent to the sum of user data, plus noise. The noise is distributed as $\Bin(n,p)$, where $p$ is chosen so that there is sufficient variance to ensure $(\eps,\delta)$-differential privacy. We take advantage of the fact that the binomial distribution is bounded: if the sum of the data is zero, the noisy sum is \emph{never} more than $n$. Hence, the analyzer will perform truncation when the noisy sum is small. We complete our proof by arguing that it is unlikely for large values to be truncated.

To streamline the presentation and analysis, we assume that $\sqrt{(100/n)\cdot\ln (2/\delta)} \leq \eps \leq 1$ so that $p \in (1/2, 1)$. We can achieve $(\eps,\delta)$ privacy for a broader parameter regime by setting $p$ to a different function; we refer the interested reader to Theorem 4.11 in \cite{CSU+19}.

\begin{figure}[ht]
\caption{The pseudocode for $\cP^\zc_{\eps,\delta}$, a private shuffled protocol for normalized binary sums}
\label{fig:zsum}

\begin{mdframed}
\textbf{Randomizer} $\cR^\zc_{\eps,\delta}(x \in \{0, 1\})$ for $\eps,\delta \in [0,1]$\textbf{:}
\begin{enumerate}
\item Let $p \gets 1 - \frac{50}{\eps^2n}\ln (2/\delta)$.

\item
    Sample $z \sim \Ber(p)$.
\item
	Output $(\underbrace{1,\dots,1}_{x+z~\textrm{copies}})$.
\end{enumerate}

\textbf{Analyzer} $\cA^\zc_{\eps,\delta}(\vec{y} \in \{1\}^*)$ for $\eps,\delta \in [0,1]$\textbf{:}
\begin{enumerate}
\item Let $p \gets 1 - \frac{50}{\eps^2n}\ln (2/\delta)$.
\item
	Let $c^* = \frac{1}{n}\cdot |\vec{y}|$, where $|\vec{y}|$ is the length of $\vec{y}$.
\item
	Output $\begin{cases}
	c^*-p & \text{if $c^* > 1$}\\
	0 & \text{otherwise}
	\end{cases}$.
\end{enumerate}
\end{mdframed}
\end{figure}

\begin{thm}\label{thm:zerocount}
For any $\eps, \delta \in [0,1]$ and any $n \in \N$ such that $n \ge (100/\eps^2) \cdot \ln(2/\delta)$, the protocol $\cP_{\eps,\delta}^\zc = (\cR^\zc_{\eps,\delta}, \cA^\zc_{\eps,\delta})$ has the following properties:
\begin{enumerate}[label=\roman*.]
\item
	$\cP_{\eps,\delta}^\zc$ is $\left(\eps,\delta\right)$-differentially private in the shuffled model.
\item
	For every $\beta\ge \delta^{25}$ , the error is $|\cP_{\eps,\delta}^\zc(\vec{x})-\frac{1}{n}\sum x_i| \leq \alpha$ with probability $\geq 1-\beta$ where 
	\begin{align*}
	\alpha
	&= \frac{50}{\eps^2n} \log\frac{2}{\delta} + \frac{1}{\eps n} \cdot \sqrt{200 \log \frac{2}{\delta}\cdot\log \frac{2}{\beta} }\\
	&= O\left(\frac{1}{\eps^2n} \log\frac{1}{\delta} \right).
	\end{align*}
\item
	$\cP_{\eps,\delta}^\zc(\underbrace{(0,\dots,0)}_{n\textrm{ copies}})=0$.
\item
    Each user sends at most two one-bit messages.
\end{enumerate}
\end{thm}

\begin{proof}[Proof of Part \textit{i}]
If we let $z_i$ be the random bit generated by the $i$-th user, the total number of messages is $|\vec{y}| = \sum_{i=1}^n x_i + z_i$.
Observe that learning $|\vec{y}|$ is sufficient to represent the output of shuffler since all messages have the same value.
Thus, the privacy of this protocol is equivalent to the privacy of
\begin{align*}
\cM(\vec{x})
&= \sum_{i=1}^{n}x_i + \mathrm{Bin}(n, p)
\sim -\left(-\sum_{i=1}^{n}x_i + \mathrm{Bin}(n,1-p)\right)+n.
\end{align*}
By post-processing, it suffices to show the privacy of $\cM_\text{neg}(\vec{x}) = -\sum_{i=1}^{n}x_i + \mathrm{Bin}(n,1-p)$ where $1-p = \frac{50}{\eps^2n}\ln \frac{2}{\delta}$. Because privacy follows almost immediately from technical claims in \cite{GGK+19}, we defer the proof to Appendix \ref{apdx:privacy}.
\end{proof}

\begin{proof}[Proof of Part \textit{ii}]
Fix any $\vec{x} \in \zo^n$. For shorthand, we define $\alpha' = 2\cdot\sqrt{\frac{p(1-p)}{n}\cdot\ln(2/\beta)}$ so that $\alpha = (1-p) + \alpha'$. A Chernoff bound implies that for $\beta \ge 2e^{-np(1-p)}$, the following event occurs with probability $\ge 1-\beta$:
\begin{equation}
\label{eq:noise-concentration}
    \left|\frac{1}{n}\cdot\sum_{i=1}^n z_i - p\right| \le \alpha'
\end{equation}
The inequality $\beta \ge 2e^{-np(1-p)}$ follows from our bounds on $\eps$, $\beta$, and $n$.

The remainder of the proof will condition on \eqref{eq:noise-concentration}. If $c^*>1$, then the analyzer outputs $c^*-p$. We show that the error of $c^*-p$ is at most $\alpha'$:
\begin{align*}
\left|(c^* - p) - \frac{1}{n}\cdot\sum_{i=1}^{n}x_i\right| &= \left| \frac{1}{n}\cdot\sum_{i=1}^{n} (x_i + z_i) - p - \frac{1}{n}\cdot\sum_{i=1}^{n}x_i \right| \tag{By construction}\\
    &= \left|\frac{1}{n}\cdot\sum_{i=1}^{n}z_i - p\right|\\
    &\le \alpha' \tag{By \eqref{eq:noise-concentration}}
\end{align*}

If $c^*\leq 1$, then the analyzer will output 0. In this case, the error is exactly $\frac{1}{n}\sum x_i$. We argue that $c^* \leq 1$ implies $\frac{1}{n}\sum x_i \leq \alpha$.
\begin{align*}
1 &\geq c^*\\
    &= \frac{1}{n}\cdot\sum_{i=1}^{n} (x_i + z_i) \tag{By construction}\\
    &\geq \frac{1}{n}\cdot\sum_{i=1}^{n} x_i + p - \alpha' \tag{By \eqref{eq:noise-concentration}}
\end{align*}
Rearranging terms yields
\begin{align*}
\frac{1}{n}\cdot\sum_{i=1}^{n} x_i &\leq (1-p) +\alpha' = \alpha
\end{align*}
which concludes the proof.
\end{proof}

\begin{proof}[Proof of Part \textit{iii}]
If $\vec{x} = (0,\ldots,0)$, then $|\vec{y}|$ is drawn from $0+\Bin(n,p)$, which implies $c^* \le 1$ with probability 1. Hence, $\cP_{\eps,\delta}^\zc(\vec{x})=0$.
\end{proof}

\subsection{A Multi-Message Protocol for Histograms}
In the protocol $\cP^\hist_{\eps,\delta}$ (Figure \ref{fig:hist}), users encode their data $x_i \in [d]$ as a one-hot vector $\vec{b} \in \zo^d$. Then protocol $\cP^\zc_{\eps,\delta}$ is executed on each coordinate $j$ of $\vec{b}$. The executions are done in one round of shuffling. To remove ambiguity between executions, each message in execution $j$ has value $j$.

\begin{figure}[ht]
\caption{The pseudocode for $\cP^\hist_{\eps,\delta}$, a private shuffled protocol for histograms}
\label{fig:hist}

\begin{mdframed}
\textbf{Randomizer} $\cR^\hist_{\eps,\delta}(x \in [d])$ for $\eps,\delta \in [0, 1]$\textbf{:}
\begin{enumerate}
\item For each $j \in [d]$, let $b_j \gets \mathbbm{1}[x = j]$ and compute scalar product $\vec{m}_j \gets j\cdot \cR^\zc_{\eps,\delta}(b_j)$.
\item
	Output the concatenation of all $\vec{m}_j$.
\end{enumerate}

\textbf{Analyzer} $\cA^\hist_{\eps,\delta}(\vec{y} \in [d]^*)$ for $\eps,\delta \in [0, 1]$\textbf{:}
\begin{enumerate}
\item
	For each $j \in [d]$, let $\vec{y}_{(j)} \gets $ all messages of value $j$, then compute $\tilde{c}_j \gets \cA^\zc_{\eps,\delta}(\vec{y}_{(j)})$.
\item
	Output $(\tilde{c}_1, \dots, \tilde{c}_d)$.
\end{enumerate}
\end{mdframed}
\end{figure}

\begin{thm}
\label{thm:hist}
For any $\eps,\delta \in [0,1]$ and any $n \in \N$ such that $n \ge (100/\eps^2) \cdot \ln(2/\delta)$, the protocol $\cP^\hist_{\eps,\delta} = (\cR^\hist_{\eps,\delta}, \cA^\hist_{\eps,\delta})$ has the following properties:
\begin{enumerate}[label=\roman*.]
\item
	$\cP^\hist_{\eps,\delta}$ is $\left(2\eps,2\delta\right)$-differentially private in the shuffled model.
\item
	For every $\beta\ge \delta^{25}$, $\cP^\hist_{\eps,\delta}$ has $(\alpha, \beta)$-per-query accuracy for 
	\begin{align*}
	\alpha = O\left(\frac{1}{\eps^2n} \log\frac{1}{\delta} \right).
	\end{align*}
\item
	For every $\beta\ge n\cdot\delta^{25}$, $\cP^\hist_{\eps,\delta}$ has $(\alpha, \beta)$-simultaneous accuracy for 
	\begin{align*}
	\alpha = O\left(\frac{1}{\eps^2n} \log\frac{1}{\delta} \right).
	\end{align*}
\item
    Each user sends at most $1+d$ messages each of length $O(\log d)$.
\end{enumerate}
\end{thm}

The accuracy guaranteed by this protocol is close to what is possible in the central model: there is a stability-based algorithm with simultaneous error $O((1/(\eps n)) \cdot \ln (1/\delta))$ \cite{BNS16}.
However, in $\cP^\hist_{\eps,\delta}$, each user communicates $O(d)$ messages of $O(\log d)$ bits.
It remains an open question as to whether or not this can be improved while maintaining similar accuracy.

Because the simultaneous error of a \emph{single-message} histogram protocol is at least $\Omega((1/(\eps n)) \cdot \poly(\log d))$ \cite{CSU+19}, this protocol is also proof that the single-message model is a strict subclass of the multi-message model. This separation was previously shown by \cite{BBGN19, BBGN19-2} for the summation problem.\footnote{In particular, a private unbiased estimator for  $\sum_i x_i$ with real-valued $x_i \in [0,1]$ in the single-message shuffled model must have error $\Omega(n^{1/6})$ \cite{BBGN19} while there exists a multi-message shuffled model protocol for estimating summation with error $O(1/\eps)$ \cite{BBGN19-2}.}

\begin{proof}[Proof of Part \textit{i}]
Fix any neighboring pair of datasets $\vec{x}\sim \vec{x}\,'$.
Let $\vec{y} \gets (\cS \circ \cR^\hist_{\eps,\delta}) (\vec{x})$ and $\vec{y}\,' \gets (\cS \circ \cR^\hist_{\eps,\delta}) (\vec{x}\,')$.
For any $j \neq j'$, the count of $j$ in output of the shuffler is independent of the count of $j'$ in the output because each execution of $\cR^\zc_{\eps,\delta}$ is independent.
As in Step (1) of $\cA_{\eps,\delta}^\hist$, for $j \in [d]$, let $\vec{y}_{(j)}$ ($\vec{y}\,'_{(j)}$ resp.) be the vector of all messages in $\vec{y}$ ($\vec{y}\,'$ resp.) that have value $j$.

For any $j \in [d]$ where $c_j(\vec{x}) = c_j(\vec{x}\,')$, $\vec{y}_{(j)}$ is identically distributed to $\vec{y}\,'_{(j)}$. For each of the two $j \in [d]$ where $c_j(\vec{x}) \neq c_j(\vec{x}\,')$, we will show that the distribution of $\vec{y}_{(j)}$ is close to that of $\vec{y}\,'_{(j)}$.
Let $\vec{r},\vec{r}\,'\in\zo^n$ where $r_i = \mathbbm{1}[x_i= j]$ and $r'_i = \mathbbm{1}[x'_i= j]$.
Now,
\begin{align*}
\vec{y}_{(j)} \sim j\cdot (\cS \circ \cR^\zc_{\eps,\delta}) (\vec{r})\text{~~and~~}\vec{y}\,'_{(j)} \sim j\cdot(\cS \circ \cR^\zc_{\eps,\delta}) (\vec{r}\,').
\end{align*}
So by Theorem \ref{thm:zerocount} Part \textit{i}, for any $T \subseteq \{j\}^*$,
\begin{align*}
\Pr[\vec{y}_{(j)} \in T]
&\leq e^\eps \cdot\Pr[ \vec{y}\,'_{(j)} \in T] + \delta.
\end{align*}
$(2\eps,2\delta)$-differential privacy follows by composition.
\end{proof}

\begin{proof}[Proof of Part \textit{ii}-\textit{iii}]
Notice that the $j$-th element in the output $\tc_j$ is identically distributed with an execution of the counting protocol on the bits $b_{i,j}$ indicating if $x_i=j$. Formally, $\tc_j \sim \cP^\zc_{\eps,\delta}(\{b_{i,j}\}_{i \in [n]})$ for all $j \in [d]$.
Per-query accuracy immediately follows from Theorem \ref{thm:zerocount} Part \textit{ii}.

To bound simultaneous error, we leverage the property that when $c_j(\vec{x})=0$, the counting protocol will report a nonzero value with probability 0. Let $Q = \{j \in [d] \,:\, c_j(\vec{x}) > 0\}$ and let $\alpha$ be the error bound defined in Theorem \ref{thm:zerocount} Part \textit{ii} for the failure probability $\beta/n$.
\begin{align*}
&\Pr\left(\exists j \in [d] \text{ s.t. } |\tc_j - c_j(\vec{x})| > \alpha\right)\\
\le &\Pr\left(\exists j \in Q \text{ s.t. } |\tc_j - c_j(\vec{x})| > \alpha\right) +\Pr\left(\exists j \notin Q \text{ s.t. } |\tc_j - c_j(\vec{x})| > \alpha\right)\\
= &\Pr\left(\exists j \in Q \text{ s.t. } |\tc_j - c_j(\vec{x})| > \alpha\right) \tag{Theorem \ref{thm:zerocount} Part \textit{iii}}\\
\le &\sum_{j \in Q}\Pr\left(|\tc_j - c_j(\vec{x})| > \alpha\right)\\
\le &\sum_{j \in Q} \beta/n \tag{Theorem \ref{thm:zerocount} Part \textit{ii}}\\
\le &\beta \tag{$|Q|\leq n$}
\end{align*}
This concludes the proof.
\end{proof}

\subsection{Applications}
\label{sec:applications}

In this section, we use our histogram protocol to solve two distributional problems; one of these results implies a very strong separation in sample complexity between the non-interactive local model and the shuffled model. Both distributional problems reduce to what we call \emph{support identification}:

\begin{defn}[Support Identification Problem]
The \emph{support identification} problem has positive integer parameters $h \leq d$.
Let $D$ be a set of size $d$ and
let $\bU_H$ be the uniform distribution over any $H\subseteq D$. The set of problem instances is $\{\bU_H \,:\, H \subseteq D\text{ and }|H| = h\}$.
A protocol solves the $(h,d)$-\emph{support identification problem with sample complexity} $n$ if, given $n$ users with data independently sampled from any problem instance $\bU_H$, it identifies $H$ with probability at least 99/100.
\end{defn}

We now show how to solve this problem in the shuffled model.
\begin{clm}
\label{clm:support-id}
Fix any $\eps\in(0,1]$ and $\delta < (1/200h)^{1/25}$. Under $(\eps,\delta)$-differential privacy, the sample complexity of the $(h,d)$-support identification problem is $O(h \log h \cdot (1/\eps^2)\cdot \log (1/\delta))$ in the shuffled model.
\end{clm}
\begin{proof}
For the purposes of this proof, we assume there is some bijection $f$ between $D$ and $[d]$ so that any reference to $j\in [d]$ corresponds directly to some $f(j)\in D$ and vice versa. Consider the following protocol: execute $\cP^\hist_{\eps,\delta}$ on $n$ samples from $\bU_H$ and then choose the items whose estimated frequencies are at least $(t+1)/n$ (the magnitude of $t$ will be determined later). We will prove that the items returned by the protocol are precisely those of $H$, with probability at least $99/100$.

Let $E_\mathrm{samp}$ be the event that every element in support $H$ has frequency at least $(2t+1)/n$ in the sample. Let $E_\mathrm{priv}$ be the event that the histogram protocol estimates the frequency of every element in $D$ with error at most $t/n$. If both events occur, every element in $H$ has estimated frequency at least $(t+1)/n$ and every element outside $H$ has estimated frequency at most $t/n$. Hence, it suffices to show that $E_\mathrm{samp}$ and $E_\mathrm{priv}$ each occur with probability $\geq 199/200$.

We lower bound the probability of $ E_\mathrm{samp}$ via a coupon collector's argument. That is, if we have $n = O(kh \log h)$ samples from $\bU_H$ then each element of $H$ appears at least $k$ times with probability at least $199/200$. Hence we set $k=(2t+1)$.

To lower bound the probability of $E_\mathrm{priv}$, we simply invoke Theorem \ref{thm:hist}: given that $\eps\in(0,1]$ and $\delta>(1/200h)^{1/25}$, the frequency of every item in $D$ is estimated up to error $t/n$ for some $t=O((1/\eps^2)\cdot \log(1/\delta))$ with probability $\geq 199/200$.\footnote{The bound on $\delta$ in Theorem \ref{thm:hist} is a function of $n$. This is derived from a pessimistic bound on the number of unique values in the input. But in this reduction, we know that data takes one of $h$ values.}
\end{proof}

In the above analysis, if we had used a protocol with simultaneous error that depends on the domain size $d$, then $t$ would in turn depend on $d$.
For example, using the histogram protocol in \cite{CSU+19} would give $t=\Omega((1/\eps) \cdot \sqrt{\log d \cdot \log (1/\delta)})$.
This results in a protocol whose sample complexity grows with $d$ in addition to $h$.

So having shown how to solve the support identification problem with few samples, we now describe two different problems and explain how to reduce these to support identification. This will imply low sample complexity in the shuffled model.

\begin{defn}[Pointer-Chasing Problem \cite{JMR19}]
The \emph{pointer chasing problem} is denoted $\mathrm{PC}(k,\ell)$ where $k,\ell$ are positive integer parameters. A problem instance is $\bU_{\{(1,a),(2,b)\}}$ where $a,b$ are permutations of $[\ell]$. A protocol \emph{solves $\mathrm{PC}(k,\ell)$ with sample complexity} $n$ if, given $n$ independent samples from any $\bU_{\{(1,a),(2,b)\}}$, it outputs the $k$-th integer in the sequence $a_1,b_{a_1},a_{b_{a_1}}\dots$ with probability at least $99/100$. 
\end{defn}

To solve $\mathrm{PC}(k,\ell)$, note that it suffices to identify  $\{(1,a),(2,b)\}$ and directly perform the pointer chasing. Because the support has size $h=2$, $\cP^\hist_{\eps,\delta}$ can be used to solve the problem with just $O((1/\eps^2)\cdot \log(1/\delta))$ samples, independent of $k$ and $\ell$. But in the case where $k=2$, \cite{JMR19} gives a lower bound of $\Omega(\ell / e^\eps)$ for non-interactive local protocols. So there is an arbitrarily large separation between the non-interactive shuffled and non-interactive local models (Theorem \ref{thm:sep-non-int}).

\begin{defn}[Multi-Party Pointer Jumping Problem \cite{JMNR19}]
The \emph{multi-party pointer jumping problem} is denoted $\mathrm{MPJ}(s, h)$ where $s,h$ are positive integer parameters. A problem instance is $\bU_{\{Z_1,\dots, Z_h\}}$ where each $Z_i$ is a labeling of the nodes at level $i$ in a complete $s$-ary tree. Each label $Z_{i,j}$ is an integer in $\{0,\dots,s-1\}$. The labeling implies a root-leaf path: if the $i$-th node in the path has label $Z_{i,j}$, then the $(i+1)$-st node in the path is the $(Z_{i,j})$-th child of the $i$-th node. A protocol \emph{solves $\mathrm{MPJ}(s,h)$ with sample complexity} $n$ if, given $n$ samples from any $\bU_{\{Z_1,\dots, Z_h\}}$, it identifies the root-leaf path with probability at least $99/100$.
\end{defn}

As with pointer-chasing, we can solve $\mathrm{MPJ}(s,h)$ when the support is identified. This takes $O(h\log h \cdot (1/\eps^2)\cdot \log (1/\delta))$ samples in the shuffled model. But \cite{JMNR19} gives a lower bound of $\Omega(h^3 / (\eps^2 \log h))$ in the local model when $s=h^4$, even allowing for sequential interactivity. However, we do not claim a polynomial separation between the shuffled model and sequentially interactive local model. This would require a proof that every sequentially interactive local protocol has a counterpart in the shuffled model.

\medskip


Note that the reductions we employ can also be applied in the central model.
That is, instead of executing $\cP^\hist_{\eps,\delta}$ in the reduction (Claim \ref{clm:support-id}), execute the central model algorithm, from \cite{BNS16}, with simultaneous error $O((1/(\eps n))\cdot \log (1/\delta))$.
This improves the bounds by $1/\eps$.

\begin{table}[H]
\caption{The sample complexity of private pointer-chasing (PC) and multi-party pointer jumping (MPJ). Shuffled and central results follow from a reduction to histograms.\vspace{1.5ex}}
\label{tab:hist-applications}
\centering
\renewcommand{\arraystretch}{1}
    \begin{tabular}{ccc}
     Model & $\mathrm{PC}(k, \ell)$ & $\mathrm{MPJ}(s, h)$ \\ \Xhline{2.5\arrayrulewidth}
     \rule{0ex}{3ex}\multirow{2}{*}{Local} & $\Omega(\ell / e^\eps)$ \cite{JMR19} & $\Omega(h^3 / (\eps^2 \log h))$ \cite{JMNR19}\\
     \vspace{.5ex}& {\footnotesize for $k=2$ } &{\footnotesize for $s=h^4$, seq. interactive }\\ \hline
     \rule{0ex}{3.5ex}\vspace{1ex}\multirow{2}{*}{Shuffled} & \multirow{2}{*}{$O\!\left(\frac{1}{\eps^2} \log \frac{1}{\delta}\right)$ } & $ O\!\left(h \log h \cdot \frac{1}{\eps^2} \log \frac{1}{\delta}\right)$  \\
     \vspace{.5ex}& &{\footnotesize for $\delta < (1/200h)^{1/25}$}\\ \hline
     \rule{0ex}{3.5ex}\vspace{1ex}Central & $O\!\left(\frac{1}{\eps}\log \frac{1}{\delta}\right)$  & $O\!\left(h\log h \cdot \frac{1}{\eps} \log \frac{1}{\delta}\right)$  \\ \Xhline{2.5\arrayrulewidth}
     \end{tabular}
\end{table}

\section{Pure Differential Privacy in the Shuffled Model}
\label{sec:pure}

In this section, we prove that any single-message shuffled protocol that satisfies $\eps$-differential privacy can be simulated by a local protocol under the same privacy constraint.

\begin{thm}[Formalization of Thm. \ref{thm:need-delta}]
\label{thm:need-delta-formal}
For any single-message shuffled protocol $\cP=(\cR,\cA)$ that satisfies $\eps$-differential privacy, there exists a local protocol $\cP_L=(\cR_L,\cA_L)$ that satisfies $\eps$-differential privacy and $\cP_L(\vec{x})$ is identically distributed to $\cP(\vec{x})$ for every input $\vec{x}\in \cX^n$.
\end{thm}

We start with the following claim, which strengthens a theorem in \cite{CSU+19} for the special case of pure differential privacy in the shuffled model:

\begin{clm}
\label{clm:pure-removal-single}
Let $\cP=(\cR,\cA)$ be any single-message shuffled protocol that satisfies $\eps$-differential privacy. Then $\cR$ is an $\eps$-differentially private algorithm.
\end{clm}
\begin{proof}
Assume for contradiction that $\cR$ is not $\eps$-differentially private.
So there are values $x,x'\in\cX$ and a set $Y\subseteq \cY$ such that 
\begin{align*}
\Pr[\cR(x)\in Y] > e^\eps\cdot \Pr[\cR(x')\in Y].
\end{align*}
Let $\vec{x} = (\underbrace{x,\dots,x}_{n~\textrm{copies}})$ and $\vec{x}\,' = (x', \underbrace{x,\dots,x}_{n-1~\textrm{copies}})$.
Now consider $Y^n$, the set of message vectors where each message belongs to $Y$.
\begin{align*}
\Pr[(\cS \circ \cR)(\vec{x}) \in Y^n]
&= \Pr[\cR(\vec{x}) \in Y^n]\\
&= \Pr[\cR(x) \in Y]^n\\
&> e^\eps\cdot\Pr[\cR(x') \in Y]\cdot\Pr[\cR(x) \in Y]^{n-1}\\
&= e^\eps\cdot\Pr[(\cS \circ \cR)(\vec{x}\,') \in Y^n]
\end{align*}
which contradicts the fact that $\cS \circ \cR$ is $\eps$-differentially private.
\end{proof}

Now we are ready to prove Theorem \ref{thm:need-delta-formal}.
\begin{proof}[Proof of Theorem \ref{thm:need-delta-formal}]
Consider the aggregator $\cA_L$ that applies a uniformly random permutation to its input and then executes $\cA$. Then $\cP_L = (\cR, \cA_L)$ is a local protocol that simulates $\cP$, in the sense that $\cP_L(\vec{x})$ is identically distributed to $\cP(\vec{x})$ for every $\vec{x}\in \cX^n$. And by Claim \ref{clm:pure-removal-single}, the randomizer is $\eps$-differentially private.
\end{proof}

\subsection{Roadblocks to Generalizing Theorem \ref{thm:need-delta-formal}}
One might conjecture Claim \ref{clm:pure-removal-single} also holds for multi-message protocols and thus immediately generalize Theorem \ref{thm:need-delta-formal}. However, this is not the case:
\begin{clm}
\label{clm:pure-removal-multi}
There exists a multi-message shuffled protocol that is $\eps$-differentially private for all $\eps \geq 0$ but its randomizer is not $\eps$-differentially private for \emph{any} finite $\eps$.
\end{clm}

\begin{proof}
Consider the randomizer $\cR^\infty$ that on input $x \in \zo$ outputs two messages $x$ and $1-x$.
The output of the shuffler $S\circ\cR^\infty$ is 0-differentially private since for all inputs the output is a random permutation of exactly $n$ 0s and $n$ 1s.
However, $\cR^\infty$ is not $\eps$-differentially private for any finite $\eps$ as the first message of $\cR^\infty(x)$ is that user's bit $x$.
\end{proof}

We note that it is without loss of accuracy or privacy to suppose that a randomizer shuffles its messages prior to sending them to the shuffler. We call these \emph{pre-shuffle} randomizers. Observe that the pre-shuffle version of $\cR^\infty$ (i.e. $\cS \circ \cR^\infty$ for 1 user) satisfies 0-differential privacy. So one might conjecture Claim \ref{clm:pure-removal-single} holds for pre-shuffle randomizers and thus generalize Theorem \ref{thm:need-delta-formal}. But this too is not the case:

\begin{clm}
There exists a multi-message shuffled protocol that is $\eps$-differentially private for some finite $\eps$ but its pre-shuffle randomizer is not $\eps$-differentially private for \emph{any} finite $\eps$.
\end{clm}
\begin{proof}

Consider any randomizer $\cR^\textrm{gap}$ that takes binary input and outputs four binary messages with the following constraint: the messages can take any value when the input is 0 but on input 1, there cannot be exactly two 1s. Formally, the supports are $\mathrm{supp}(\cR^\mathrm{gap}(0)) = \zo^4$ and $\mathrm{supp}(\cR^\mathrm{gap}(1)) = \{\vec{y} \in \zo^4 \,:\, \sum_{i}y_i \neq 2\}$.

The pre-shuffle randomizer $\cS \circ \cR^\textrm{gap}$ cannot satisfy pure differential privacy because
$(0,0,1,1) \in \mathrm{supp}(\cR^\mathrm{gap}(0))$ but $(0,0,1,1) \notin \mathrm{supp}(\cR^\mathrm{gap}(1))$.
On the other hand, for all $n\ge 2$ and $\vec{x} \in \zo^n$,
\begin{align*}
\mathrm{supp}(\cS\circ\cR^\mathrm{gap}(\vec{x})) = \zo^{4n}.
\end{align*}
This follows from the fact that every number in $\{0, \ldots, 4n\}$---the number of 1s sent to the shuffler---can be expressed as the sum of $n$ numbers from $\{0, 1, 3, 4\}$.
Thus, there is some finite $\eps$ for which the protocol with randomizer $\cR^\mathrm{gap}$ is $\eps$-differentially private.
\end{proof}

\ifnum\arxiv=1
    \section*{Acknowledgments}
    We are grateful to Daniel Alabi and Maxim Zhilyaev for discussions that shaped the early stages of this work.
    We are also indebted to Matthew Joseph and Jieming Mao for directing us to the pointer-chasing and multi-party pointer-jumping problems. Finally, we thank Salil Vadhan for editorial comments and providing a simpler construction for Claim \ref{clm:pure-removal-multi}.
\fi
\bibliography{shuffle}

\appendix

\section{Privacy via Smooth Distributions}
\label{apdx:privacy}
Ghazi, Golowich, Kumar, Pagh and Velingker \cite{GGK+19} identify a class of distributions and argue that, if $\eta$ is sampled from such a distribution, adding $\eta$ to a 1-sensitive sum ensures differential privacy of that sum.

\begin{defn}[Smooth Distributions, \cite{GGK+19}]
A distribution $\bD$ over $\Z$ is $(\eps,\delta,k)$-smooth if for all $k' \in [-k, k]$,
\[
\Pr_{Y\sim \bD}\left[\frac{\Pr_{Y'\sim \bD}[Y'=Y]}{\Pr_{Y'\sim \bD}[Y'=Y+k']} \geq e^{|k'|\eps}\right]\leq \delta.
\]
\end{defn}

\begin{lem}[Smoothness for Privacy, \cite{GGK+19}]
\label{lem:smooth-for-privacy}
Let $f:\Z^n \rightarrow \Z$ be a function such that $|f(\vec{x}) - f(\vec{x}\,')| \le 1$ for all $\vec{x}\sim\vec{x}\,'$.
Let $\bD$ be an $(\eps,\delta,1)$-smooth distribution. The algorithm that takes as input $\vec{x} \in \Z^n$, then samples $\eta \sim \bD$ and reports $f(\vec{x}) + \eta$ satisfies $(\eps,\delta)$-differential privacy.
\end{lem}

\begin{lem}[Binomial Distribution is Smooth, \cite{GGK+19}]
\label{lem:binomial-is-smooth}
For any positive integer $n$, $\gamma \in [0,1/2]$, $\alpha \in [0,1]$, and any $k\leq \alpha\gamma n/ 2$, the distribution $\Bin(n,\gamma)$ is $(\eps, \delta, k)$-smooth with
\begin{align*}
\eps
= \ln \frac{1+\alpha}{1-\alpha}
~~~~~\text{and}~~~~~
\delta
= \exp\left(-\frac{\alpha^2\gamma n}{8}\right) + \exp\left(-\frac{\alpha^2\gamma n}{8 + 2\alpha}\right).
\end{align*}
\end{lem}

\begin{coro}
\label{coro:mneg-private}
Fix any $\eps,\delta \in [0,1]$.
Let $n \ge (100/\eps^2)\cdot\ln(2/\delta)$.
The algorithm $\cM_\text{neg}$ that takes as input $\vec{x} \in \{0,-1\}^n$ then samples
\[
\eta \sim \Bin\left(n,\,50\cdot \frac{\ln (2/\delta)}{n\eps^2}\right)
\]
and reports $\eta + \sum x_i$ satisfies $(\eps,\delta)$-differential privacy.
\end{coro}
\begin{proof}
When $\alpha = (e^\eps-1)/(e^\eps+1)$ observe that $ \alpha \in [\eps/\sqrt{5}, 1)$ and Lemma \ref{lem:binomial-is-smooth} implies that $\eta$ is sampled from an $(\eps,\delta,1)$-smooth distribution:
\begin{align*}
\ln \frac{1+\alpha}{1-\alpha} &= \ln \frac{(e^\eps+1) + (e^\eps - 1)}{(e^\eps + 1) - (e^\eps - 1)}
= \eps
\end{align*}
and
\begin{align*}
\exp\left(-\frac{\alpha^2\gamma n}{8}\right) + \exp\left(-\frac{\alpha^2\gamma n}{8 + 2\alpha}\right)
&\leq 2\exp\left( -\frac{\alpha^2\gamma n}{10}\right) \tag{$\alpha < 1$}\\
&\le 2\exp\left(-\frac{\gamma \eps^2 n}{50}\right)\\
&= \delta.
\end{align*}
So by Lemma \ref{lem:smooth-for-privacy}, we have $\cM_\text{neg}$ is $(\eps,\delta)$-differentially private.
\end{proof}

\end{document}